\documentclass[preprint,12pt]{elsarticle}
\usepackage[utf8]{inputenc}
\usepackage[T1]{fontenc}     
\usepackage[english]{babel}  

\usepackage{amssymb}
\usepackage{amsmath}

\usepackage{graphicx}
\usepackage{xcolor}
\usepackage{color}
\usepackage{hyperref}
\usepackage{geometry}
\usepackage{subcaption}
\usepackage{comment}
\usepackage{amsthm}

\newcommand{\wt}[1]{\widetilde{#1}}

\newcommand{\ber}[2]{\text{ber}_{#1}\left(#2\right)}
\newcommand{\bei}[2]{\text{bei}_{#1}\left(#2\right)}

\newtheorem{proposition}{Proposition}

\journal{Applied Mathematical Modelling}

\begin{document}

\begin{frontmatter}



\title{On mathematical characterization of a Bessel functions-based passive element in electronic circuits\footnotetext{{\it Published in:} Applied Mathematical Modelling, 154 (2026) 116683. {\bf DOI}: \href{https://doi.org/10.1016/j.apm.2025.116683}{10.1016/j.apm.2025.116683}}} 


\author{Ivano Colombaro} 

\affiliation{organization={Faculty of Engineering, Free University of Bozen-Bolzano},
            addressline={via Bruno Buozzi 1}, 
            city={Bolzano},
            postcode={39100}, 
            country={Italy}}

\author{Marc Tudela-Pi}
\affiliation{Instituto de Microelectronica de Barcelona (IMB-CNM),  Consejo Superior de Investigaciones Cientificas (CSIC), Bellaterra, Spain}
\affiliation{CIBER-BBN, Instituto de Salud Carlos III, Madrid, Spain}

\begin{abstract}
Modeling relaxation phenomena in complex media is central to understanding multiscale dynamics in materials science, bioengineering and condensed matter physics. Existing fractional-order models, while flexible, sometimes lack physical interpretability, closed-form time-domain expressions, and compatibility with physically realizable architectures. In this work, we propose a novel passive element whose impedance and admittance are defined analytically via modified Bessel functions of first kind, through the electro-mechanical analogy. This approach preserves key physical properties such as analyticity, passivity, BIBO (bounded-input, bounded-output) stability and monotonicity, while enabling the direct use of its time-domain representation in simulations and system modeling. As an application, we demonstrate that this model accurately captures the broadband dispersive behavior of biological tissues, offering a physically grounded and tractable alternative to fractional-order formulations.
\end{abstract}


\end{frontmatter}



\section{Introduction}

A relaxation process refers to the time-dependent evolution by which a system progressively returns to equilibrium after experiencing an external perturbation \cite{jonscher1999}. These processes are fundamental to understanding how energy is dissipated, how internal structures reorganize, and how physical fields adjust in response to time-varying stimuli. Relaxation dynamics are investigated across a wide range of scientific disciplines, including polymer physics \cite{boyd2007}, viscoelasticity \cite{metzler2003,colombaro2017bessel}, dielectrics \cite{garrappa2016modelsofdielectric,oliveira2011MLrelaxation}, glassy systems \cite{narayanaswamy1971}, and biological tissues \cite{gabriel2009,ivorra2003}. In each of these contexts, the observed temporal response reflects the intrinsic material structure, the nature of transport processes, and the interaction among coupled dynamical modes.

Biological systems, in particular, exhibit structural and functional complexity that gives rise to diverse relaxation behaviors. Tissues are composed of heterogeneous elements, notably cellular membranes with ionic channels, extracellular matrices and ionically conductive media. These components influence dielectric behavior through processes spanning a broad range of spatial scales and relaxation times \cite{martinsen2011}. Consequently, their response to external perturbations often involves a wide distribution of relaxation times and non-trivial frequency-dependent characteristics that reflect the diversity of underlying physiological mechanisms \cite{gabriel1996}. To model such behavior, various theoretical and empirical approaches have been proposed. The simplest is the Debye model \cite{debye1929}, which is mathematically equivalent to a single resistor-capacitor (RC) circuit element and assumes a unique relaxation time. This leads to a purely exponential decay, which fails to capture the broad or asymmetric relaxation spectra commonly observed in heterogeneous systems \cite{martinsen2011}. In response to these limitations, generalized formulations such as the Cole-Cole \cite{cole1941}, Cole-Davidson \cite{davidson1951}, and Havriliak-Negami \cite{havriliak1967} models introduce empirical parameters to account for continuous distributions of relaxation times. Although these fractional-order models offer improved agreement with experimental data, they sometimes fail to provide a clear connection to underlying physical mechanisms. In addition, their frequency-domain expressions are not easily invertible, which complicates their use in time-domain simulations and hinders their integration into circuit-level representations \cite{Mainardi2010book}. These issues underscore the need for alternative models that are both analytically tractable and more directly connected to the physics of the system \cite{naranjo2019}.

In this paper, we introduce a novel passive electrical component whose impedance is defined in the frequency domain by a ratio of modified Bessel functions of the first kind. This element captures distributed relaxation behavior within a compact and physically interpretable formulation and it yields closed-form expressions that are invertible in the time domain, and it is fully compatible with circuit-level modeling. These properties make it particularly suitable for describing dielectric and bioelectrical responses in systems with complex relaxation dynamics.

Specifically, in Sect.~\ref{sec:visco} we outline how the Bessel models originally emerged from studies in linear viscoelasticity, particularly in the context of hemodynamic modeling. Moreover, we introduce the electromechanical analogy, which establishes a correspondence between viscoelastic mechanical elements, such as springs and dashpots, and their electrical counterparts, namely resistors, capacitors, and inductors, thereby providing a unified framework for representing dissipative systems across domains and so enabling the translation of mechanical relaxation phenomena into impedance-based electrical models.
Subsequently, we define and characterize an impedance based on Bessel functions in Sect.~\ref{sec:impedance}. 
From a physical standpoint, \emph{impedance} represents the dynamic opposition that a system offers to the flow of energy when subjected to a time-varying excitation. It embodies the balance between \emph{energy storage} and \emph{dissipation}, linking the effort variable (voltage, force, pressure) to the corresponding flow variable (current, velocity, flux) through a causal, frequency-dependent relation.
In electrical systems, the impedance $Z(s)$ connects voltage and current via $V(s) = Z(s) I(s)$, but the same concept extends naturally to mechanical, acoustic, or biological systems through the \emph{electromechanical analogy}, where $Z(s)$ quantifies how the system resists or delays the transfer of energy.
Interestingly enough, Bessel impedance  satisfies the properties such as analiticity, differentiability, causality, monotonicity and BIBO stability, namely bounded input produces a bounded output, ensuring that the system does not exhibit unbounded behavior in response to finite excitations.
Such impedance properties imply that the system is passive, where a passive element is a system component that does not generate energy but only stores or dissipates it. 
Indeed, a \emph{passive impedance} is one that cannot deliver net energy to the surroundings, ensuring stability and consistency with the second law of thermodynamics.
The proposed Bessel impedance follows this framework as a passive element whose impedance describes non-exponential relaxation and distributed memory, extending the concept of diffusive or fractional impedances to Bessel-based dynamics.
After dealing with the impedance, which admits also an analytical time-domain representation, we illustrate the admittance and its properties in Sect.~\ref{sec:admittance}.
Admittance is the reciprocal of impedance and represents the ease with which a system allows energy to flow in response to an external excitation. 
In the next Sect.~\ref{sec:plots}, we then reproduce the plots of the magnitude and the phase of the impedance of the passive element based on Bessel functions, in particular concerning the application to biological tissues and their physical interpretation.
A discussion on the topic is finally written in Sect.~\ref{sec:discussion} and concluding remarks are in Sect.~\ref{sec:end}.

\section{Viscoelasticity and electro-mechanical analogy}
\label{sec:visco}

Bessel models in linear viscoelasticity  have been applied to problems in hemodynamic, particularly in modeling pulse wave propagation in fluid-filled elastic tubes, such as arteries~\cite{giusti2016dynamic}. In particular, they have been formalized~\cite{colombaro2017bessel} and applied~\cite{colombaro2017oneparameter} to viscoelastic models, also by with defining the so-called specific attenuation factor~\cite{colombaro2023Qbessel} and by introducing the wave-front approximation~\cite{colombaro2017onthepropagation}, but there are applications also in other fields, e.~g. modeling acoustic impedance~\cite{drozda2023diffusive}.
From the physical point of view, the relevance of these models in linear viscoelasticity is due to their extension of classical Maxwell model, since they exhibit fractional Maxwell-like behavior at short times and standard Maxwell behavior at long times, making them versatile for modeling both fluid-like and solid-like responses.

In the context of linear viscoelasticity, one of the main expressions is the reciprocity relation between the material functions in the Laplace domain~\cite{Mainardi2010book}. 
Given the creep compliance $J(t)$ and the relaxation modulus $G(t)$, in a viscoelastic Bessel medium the reciprocity relation is written in terms of the modified Bessel functions of the first kind $I_\alpha$, namely 
\begin{equation}
    s\wt{J}(s) = \frac{I_{\nu}(\sqrt{s})}{I_{\nu+2}(\sqrt{s})} = \frac{1}{s\wt{G}(s)}, \qquad \nu >-1 \,,
\end{equation}
reminding that the series representation of these special functions is
\begin{equation} \label{eq:I_nu}
I_{\alpha} (z) := \left( \frac{z}{2} \right) ^{\alpha} \sum _{m=0} ^\infty \frac{1}{m! \, \Gamma (m+\alpha+1)} \left( \frac{z}{2} \right) ^{2m} \,,
\end{equation}
where $\Gamma(z)$ is Euler Gamma function~\cite[Ch.~6]{Abramowitz1965handbook}.

There exists a recognized equivalence between viscoelastic systems and electrical ladder structures, as first formalized by Gross and Fuoss~\cite{gross1956electricalanalogs, gross1956ladderstructures} and subsequently developed by Giusti and Mainardi~\cite{giusti2016infiniteseries} to show how a viscoelastic response function can also be interpreted as the response of an equivalent electrical circuit, by exploiting Laplace transforms and Dirichlet series. This analogy establishes a direct correspondence between the physical quantities involved: {the electric potential $V$ is analogous to the mechanical stress $\sigma$, and the electric current $i$ corresponds to the strain rate $\dot{\varepsilon}$. In this mapping, resistor $R$ reflects viscosity, while capacitor $C$ represents the reciprocal of the elastic modulus}. This electro-mechanical analogy enables the use of electrical ladder networks to represent viscoelastic models exhibiting relaxation behavior. Such structures are particularly suitable for capturing distributed relaxation times while remaining compatible with realizable passive circuit architectures. Unlike fractional-order models, which often require abstract or non-physical elements, this approach leads to implementable models that are well suited for both simulation and hardware realization. 

The electro-mechanical analog for this class of viscoelastic models based on modified Bessel functions is derived in~\cite{colombaro2018besselwseas}, and within this framework we can write the following relation between the electric current \( \widetilde{i}(s) \) and the electric potential \( \widetilde{V}(s) \) in the Laplace domain as

\begin{equation}\label{eq:V/I-Lap}
    \frac{\wt{V}(s)}{\wt{i}(s)} = \frac{I_{\nu}(\sqrt{s})}{I_{\nu+2}(\sqrt{s})} \,,\qquad \nu > -1 \, .
\end{equation}

As a result of this formal correspondence, we define an electrical component based on the Bessel medium and depending so on the parameter coming from the order $\nu$ of Bessel functions, namely
\begin{equation}\label{eq:Bstor}
    {B}_\nu(s) := \frac{I_{\nu}(\sqrt{s})}{I_{\nu+2}(\sqrt{s})} \,,\qquad \nu > -1 \, .
\end{equation}
This approach is motivated by the fact that from the viscoelastic point of view this class of models has a physical interpretation and it matches a fractional Maxwell model of order $1/2$ for short times together with an ordinary Maxwell model at long time. So, in the field of electronic engineering, it can contribute to enrich the literature and to describe a new class of models, going beyond the classical ones.
Hereinafter, we are going to develop the properties of this original element, highlighting strengths coming from the use of this approach, that is able to describe a resistive and capacity behavior at the same time.

\section{Definition and properties of the Bessel element impedance}
\label{sec:impedance}


At this point, let us consider a circuit composed by a voltage source $V$ and a Bessel element and let us deduce its characteristics, by means of impedance $\wt{Z}_B$, which we properly define in terms of physical units. In the Laplace domain, the impedance has to be proportional to the ratio of the two modified Bessel functions, namely $\wt{Z}_B\sim B_\nu (s)$, and it is convenient in these cases to have an dimensionless argument, so that we introduce a relaxation time $\tau$ and we compute the following rescaling
\begin{equation}
    \wt{Z}_B \sim B_\nu (s\tau) = \frac{I_{\nu}(\sqrt{s\tau})}{I_{\nu+2}(\sqrt{s\tau})} \,,\qquad \nu > -1 \, .
\end{equation}
At this stage, to have the impedance measured in Ohm $(\Omega)$, we multiply the rescaled nondimensional $B_\nu$ by a resistor $R_\infty$, and we obtain the final formula for the impedance of a novel passive element depending on Bessel functions, written
\begin{equation}\label{eq:ZBdef}
    \wt{Z}_B (s) = R_\infty \frac{I_{\nu}(\sqrt{s\tau})}{I_{\nu+2}(\sqrt{s\tau})} \,,\qquad \nu > -1 \, ,
\end{equation}
so that, by definition
\begin{equation}\label{eq:ZdefLaplace}
    \frac{\wt{V} (s)}{\wt{i} (s)} = \wt{Z}_B (s) \,.
\end{equation} 


From the mathematical point of view, an element as~\eqref{eq:ZBdef} presents properties that are crucial in mathematical modeling of physical phenomena.
First of all, Bessel functions are by definition analytic for $\nu>-1$~\cite{baricz_2010Bessel}, being defined as a power series in~\eqref{eq:I_nu}, they converge absolutely and uniformly on compact subsets of $\mathbb{C}$.

\begin{proposition}
Let $\nu > -1$ and $\tau > 0$ be a fixed positive real parameter.
Then, Define the function
$$ B_\nu (s\tau) := \frac{I_{\nu}(\sqrt{s\tau})}{I_{\nu+2}(\sqrt{s\tau})} \,,\qquad \nu > -1 \,, $$
where $I_\nu(z)$ denotes the modified Bessel function of the first kind of order $\nu$. Then $B_\nu(s\tau)$ is analytic on the domain
$$\mathcal{D} = \left\{ s \in \mathbb{C} \setminus (-\infty, 0] \,\middle| \, I_{\nu+2}(\sqrt{s\tau}) \neq 0 \right\} \,.$$
\end{proposition}
\begin{proof}
The modified Bessel function $I_\nu(z)$ is entire for every $\nu > -1$, as follows from its power series representation~\eqref{eq:I_nu}, which converges absolutely for all $z \in \mathbb{C}$.
The substitution $z = \sqrt{s\tau}$ is analytic on $\mathbb{C} \setminus (-\infty, 0]$ due to the branch cut of the complex square root. Therefore, both $I_\nu(\sqrt{s\tau})$ and $I_{\nu+2}(\sqrt{s\tau})$ are analytic on this domain.

The quotient of two analytic functions is analytic wherever the denominator does not vanish. Since the zeros of $I_{\nu+2}(z)$ are isolated and do not lie on the positive real axis for $\nu > -1$, the function $B_\nu(s\tau)$ is analytic on $\mathcal{D}$.

\end{proof}

Analyticity implies that functions are infinitely differentiable and it leads to a very strong form of smoothness, so that the function $B_\nu$, and so $\wt{Z}_B$, exhibits smooth and continuous differentiability~\cite{simpson1984modifiedBessel, segura2021modifiedBessel}, ensuring gradual transitions across the domain $\mathcal{D}$.

Having an impedance function $\wt{Z}_B(s)$ analytic in the right half-plane $\Re\{s\}>0$ is a necessary condition for causality. The necessary condition is given by the fact of having vanishing inverse Laplace transform for $t<0$, as will be developed in details in~\ref{subsec-timedomain}, but it can be straightforwardly deduced by the analytic expression in~\eqref{eq:ZB-timedomain}.
Due to these facts, we can conclude that the system is \textit{causal}. 

Causality implies that response can not be produced before an input, and the system is \textit{stable}, guaranteeing that bounded input leads to bounded output, known as BIBO stability. A causal and stable function that is positive real in $\Re\{s\}>0$ identifies a \textit{passive} element, that is relevant in this context since a system is said {passive} if it does not generate energy, but it can only store or dissipate it, supporting physical realism.

It has been observed for the viscoelastic analogous model~\cite{colombaro2017bessel} that, for all values of $\nu$, function expressed as $B_\nu$ in the Laplace domain decreases monotonically in the time domain, so that the same holds for $\wt{Z}_B(s)$. The property of monotonicity is relevant because monotonic impedance functions tends to have a meaningful behavior from the physical point of view. Specifically, monotonicity ensures uniqueness of solution in inverse problems and it often aligns with passive system behavior, especially when combined with analyticity and positive realness, as in our case.
Moreover, from the numerical point of view, monotonic functions are more stable and they reduce the risk of local minima in optimization algorithms used for fitting.

It might also be interesting to study the asymptotic behavior of $B_\nu$. For
vanishing argument and fixed $\nu>-1$, we know that the  the series expansion of modified Bessel functions tends to~\cite[Eq.~(9.6.7)]{Abramowitz1965handbook}
\begin{equation}
    I_{\nu}(z) \overset{{z\to 0}}{\sim} \frac{1}{\Gamma(\nu+1)}\left(\frac{z}{2}\right)^\nu \,,
\end{equation}
leading to
\begin{equation}\label{eq_Bnu-asym0}
    B_\nu (s\tau) \overset{{s\to 0}}{\sim} \frac{\Gamma(\nu+3)}{\Gamma(\nu+1)} \frac{4}{s\tau} \,,
\end{equation}
so proportional to $1/s$.
Differently, the modified Bessel functions for large arguments are computed in accordance with a Tauberian theorem by means of the inverse Laplace transform, as shown in~\cite[Eq.~(35)]{colombaro2017bessel} or more detailed in~\cite[Appendix~A]{giusti2016infiniteseries}, and it is
\begin{equation}
     I_{\nu}(z) \overset{{\vert z \vert \to \infty}}{\sim} \frac{e^z}{\sqrt{2\pi z}} \,,
\end{equation}
giving then the constant result
\begin{equation}\label{eq_Bnu-asyminfty}
    B_\nu (s) \overset{{s\to \infty}}{\sim} 1 \,.
\end{equation}
To summarize, we can write the asymptotic behavior of $\wt{Z}_B$ for low and high frequency, respectively
\begin{gather}
    \wt{Z}_B (s) \overset{{s\to 0}}{\sim} \frac{\Gamma(\nu+3)}{\Gamma(\nu+1)} \frac{4\, R_\infty}{s\tau} \,, \label{eq:asymzeroZ} \\
    \wt{Z}_B (s) \overset{{s\to \infty}}{\sim} R_\infty \,. \label{eq:asyminftyZ}
\end{gather}
As a more technical note, the real part of $\widetilde{Z}_B(j\omega)$ exhibits a globally decreasing profile, while the asymptotic behavior shows high dissipation at low frequencies, with $\Re[\widetilde{Z}_B(j\omega)] \sim \omega^{-1}$ as $\omega \to 0$, and convergence to a constant \( R_\infty \) as \( \omega \to \infty \). The absence of oscillatory terms and the smoothness of the Bessel quotient suggest a decay that is monotonic in practice, consistent with the expected behavior of relaxation-based systems.

In addition, the impedance function in~\eqref{eq:ZBdef} can be approximated using finite series expansions. In details, \textit{approximability via finite series expansion} refers to the ability to represent a complex function as a truncated sum of simpler, well-understood basis functions, typically polynomials or exponential terms or special functions, like Bessel functions. This property is foundational in both theoretical and computational contexts, as it enables the transformation of infinite-dimensional problems into finite-dimensional approximations that are tractable for analysis and numerical computation. The convergence of such series is often governed by means of analyticity or smoothness of the target function, and that the approximation error can be made arbitrarily small by increasing the number of terms. Moreover, the structure of the series often encodes physical insights, such as diffusive or memory effects, making approximability not only a computational convenience, but also a bridge between abstract mathematics and real-world phenomena.
In our case, the concept of approximation via finite series expansion, combined with monotonicity, analyticity, and passivity of the Bessel element, makes it a powerful and elegant model for bioimpedance, since finite series expansions allow for fast and accurate numerical evaluation, which is crucial then in real-time applications.

\subsection{Spectral control and parameter interpretability}

Spectral dispersion usually refers to the variation of a system's response with respect to the frequency~\cite{femmam2017fundamentals}. Thus, a smooth spectral dispersion implies that this variation is continuous and differentiable, without abrupt transitions or resonances.
This is particularly relevant in biological systems, diffusive media, or fractional-order systems, where responses are gradual rather than sharply resonant. This behavior is elegantly captured by models like the proposed element, which leverage the analytic structure of modified Bessel functions to produce impedance functions that are not only smooth and monotonic but also approximable via finite series expansions.

Moreover, the smoothness and analytic nature of the impedance $\wt{Z}_B$ enable spectral control, namely the ability to shape the frequency response predictably and continuously, making it ideal for applications in bioimpedance spectroscopy, tissue characterization, and circuit synthesis. Together, these properties form a powerful modeling framework that bridges mathematical elegance with practical utility.

From~\eqref{eq:ZBdef}, we notice that $\wt{Z}_B(s)$ depends on $s$, the variable in the Laplace domain related to the angular frequency $\omega$ thanks to the relation $s=\jmath \omega$, which technically allows then to move to Fourier transform, and by three parameters: $R_\infty$, $\tau$, and $\nu$.
An attentive reader might have deduced that we refer to $R_\infty$ as the value of the impedance for high frequencies, measured in Ohm $(\Omega)$, namely $s\to\infty$, as a consequence of~\eqref{eq_Bnu-asyminfty}. 
Instead, $\tau$ is the relaxation time, measured then in seconds $(s)$, that controls the frequency scaling of the impedance response. As a matter of fact, $s\tau \ll 0$ identifies low-frequency behavior (so long timescales), while $s\tau \gg 0$ means high-frequency behavior (or short timescales).
Finally, $\nu$ is a dimensionless parameter coming from modified Bessel functions, that shapes the spectral dispersion and it controls the smoothness and curvature of the impedance function across frequencies. We will also deepen hereinafter that there is a functional interdependence between $\nu$ and $\tau$, respectively the order and a part of the argument $z=\sqrt{s\tau}$ of Bessel functions, as might be intuitively deduced from the series representation~\eqref{eq:I_nu}. Focusing on the impedance $\wt{Z}_B(s)$, increasing $\nu$ shifts the weight of the series to higher powers of $\sqrt{s\tau}$, concentrating so the response in a narrower frequency band, while increasing $\tau$ stretches the argument $\sqrt{s\tau}$, effectively shifting the frequency response.
We might appreciate from the following pages that, in the electromechanical analogy, $\nu$ serves as a tuning parameter that modulates the interplay between elastic energy storage and viscous dissipation, thereby shaping the frequency dependence of the system. Rather than representing a single physical quantity, such as mass or resistance, it encapsulates the intricate internal relaxation dynamics in a concise and adjustable form.

\subsection{Time-Domain Representation}\label{subsec-timedomain}

Unlike fractional-order models, whose time-domain kernels often lack closed-form expressions or require numerical inversion, the Bessel impedance admits an exact analytical formulation in the time domain. This is particularly relevant for simulating transient responses, implementing time-stepping schemes, or analyzing memory effects in physical systems.

In particular, the inverse Laplace transform of the transfer function $B_\nu(s)$ in~\eqref{eq:Bstor} yields
\begin{equation}
    B_\nu(t) = \frac{2(\nu + 1)}{\tau} + \frac{4(\nu + 1)(\nu + 2)}{\tau} \sum_{k=1}^{\infty} \exp\left( - j^2_{\nu+2,k} \frac{t}{\tau} \right), \qquad \nu > -1 \,,
\end{equation}
where $j_{\nu+2,k}$ denotes the $k$-th positive zero of the Bessel function $J_{\nu+2}(z)$.
Consequently, the impulse response of the impedance $\wt{Z}_B(s)$ defined in~\eqref{eq:ZBdef} is
\begin{equation}\label{eq:ZB-timedomain}
    z_B(t) = \frac{R_\infty}{\tau} \left[ 2(\nu + 1) + 4(\nu + 1)(\nu + 2) \sum_{k=1}^\infty \exp\left( - j^2_{\nu+2,k} \frac{t}{\tau} \right) \right],
\end{equation}
which has units of $\Omega/\text{s}$ and defines a causal convolution kernel
\begin{equation} \label{eq:conv_bessistor}
    V(t) = \left( z_B * i \right) (t) = \int_0^\infty z_B(t - t') i(t')\, \mathrm{d}t'\,,
\end{equation}
where $*$ refers to Laplace convolution.
Although impedance is typically a frequency-domain concept, transforming it into the time domain might allow, for instance, to analyze transient behavior of circuits or to understand causality and memory effects, apart from playing a crucial role in numerical simulations and finite-element methods.
Furthermore, several known techniques could be implemented for the computation of the inverse Laplace transform~\cite{garrappa2015numerical, garrappa2018survey}.

The representation in time domain in~\eqref{eq:ZB-timedomain} reveals the modal structure of this novel element: each exponential term corresponds to a relaxation mode with a specific decay rate. The series converges rapidly and can be truncated for practical computations. Still, for efficient simulation over long time spans or in multiscale scenarios, it is useful to compress the tail of the kernel.

\subsubsection{Hybrid kernel approximation}

While the time-domain kernel $z_B(t)$ derived in the previous section admits an exact modal expansion, the infinite summation may be computationally intensive or unnecessary in practical implementations. To address this, we introduce a \textit{hybrid kernel} formulation that preserves exact accuracy in the early-time response while employing a compressed representation for the long-time tail.

For numerical and practical purposes, the impulse response can be naturally decomposed into two contributions
\begin{equation} \label{eq:hybridkernel}
    z_B(t) = z_B^{(N)}(t) + z_B^{\text{tail}}(t),
\end{equation}
defined as
\begin{align}
    z_B^{(N)}(t) &= \frac{R_\infty}{\tau} \left[ 2(\nu + 1) + 4(\nu + 1)(\nu + 2) \sum_{k=1}^{N} \exp\left( - \frac{j_{\nu+2,k}^2 t}{\tau} \right) \right], \label{eq:zkexact} \\
    z_B^{\text{tail}}(t)  &= \frac{R_\infty}{\tau} \mathcal{G}_\nu^{(N)}(t) \,, \label{eq:zktail}
\end{align}
where $\mathcal{G}_\nu^{(N)}(t)$ approximates the residual tail contribution beyond $k = N$, we introduce a compact surrogate based on a rational-exponential expression
\begin{equation}\label{eq:GnuN}
\mathcal{G}_\nu^{(N)}(t) \approx A_\nu^{(N)} \frac{e^{- \lambda_\nu^{(N)} t / \tau}}{1 - r_\nu^{(N)} e^{- \lambda_\nu^{(N)} t / \tau}},
\end{equation}
where the parameters \( A_\nu^{(N)} \), \( r_\nu^{(N)} \), and \( \lambda_\nu^{(N)} \) are fitted to match the amplitude and decay profile of the Bessel tail beyond the truncation point.  This approximation drastically reduces the number of required terms while retaining physical interpretability and spectral continuity and, interestingly, \eqref{eq:GnuN} admits the exact expansion
\begin{equation}
\mathcal{G}_\nu^{(N)}(t) = A_\nu^{(N)} \sum_{k=0}^\infty r_\nu^{(N)k} \exp\left( - \frac{(k+1) \lambda_\nu^{(N)}}{\tau} t \right).
\end{equation}
The resulting hybrid kernel achieves a balance between interpretability and efficiency: it preserves the dominant physical modes exactly while modeling the remaining infinite hierarchy with a compact analytical form. This makes it particularly attractive for use in time-domain solvers, finite-element implementations and real-time applications. 
In details, this construction captures the key properties of the exact tail: causality, smoothness, and exponential decay. It also preserves the passive and physically realizable structure of the kernel, while dramatically reducing the computational cost of evaluating the tail. 

Crucially, the parameters \( \lambda_\nu^{(N)} \), \( r_\nu^{(N)} \), and \( A_\nu^{(N)} \) can be estimated directly from the known asymptotic behavior of the Bessel zeros. For instance, large values of $k$ give
\begin{equation}
j_{\nu+2,k} \sim \pi \left( k + \frac{\nu + 1}{2} \right), \qquad k \to \infty \,,
\end{equation}
which imply that the exponents in the tail grow approximately quadratically, leading to
\begin{align}
\lambda_\nu^{(N)} &:= \pi^2 \left( N + 1 + \frac{\nu + 1}{2} \right)^2, \\
r_\nu^{(N)} &:= \left( \frac{ N + 2 + \frac{\nu + 1}{2}}{N + 1 + \frac{\nu + 1}{2}} \right)^2, \\
A_\nu^{(N)} &:= \frac{4 (\nu + 1)(\nu + 2)}{1 - r_\nu^{(N)}}\,.
\end{align}

{Furthermore, we now show that this hybrid formulation accurately reproduces the full kernel behavior with negligible error beyond the truncation point, and significantly accelerates convolution-based simulations without compromising physical consistency.}
Let us write the hybrid kernel in~\eqref{eq:hybridkernel} as
\begin{equation}
    z_B^{(N+\text{tail})}(t) = z_B^{(N)}(t) + \frac{R_\infty}{\tau} \mathcal{G}_\nu^{(N)}(t) \,,
\end{equation}
according to~\eqref{eq:zkexact} and~\eqref{eq:zktail}.
To evaluate the accuracy of the hybrid kernel representation, we quantify the discrepancy between the exact kernel $z_B(t)$ in time domain, defined in~\eqref{eq:ZB-timedomain}, and its truncated approximations. In particular, we consider the \emph{pure modal truncation} $z_B^{(N)} (t)$, which retains only the first $N$ exponential terms, and the \emph{hybrid approximation} $z_B^{(N+\text{tail})} (t)$, which combines the exact sum of the first $N$ modes with a compressed geometric tail.
We define the relative approximation error as
\begin{equation}
\varepsilon(N) := \frac{\left\| z_B(t) - {z}_B^{*}(t) \right\|_2}{\left\| z_B(t) \right\|_2} 
\end{equation}
where \( {z}_B^{*}(t) \) stands for either \( z_B^{(N)}(t) \) or \( z_B^{(N+\text{tail})}(t) \), and the norm \( \| \cdot \|_2 \) is evaluated numerically over a fixed interval \( t \in [0, T] \), using a uniform discretizations of \( M \) time samples.
This error measure is \emph{independent of any input signal} $i(t)$ and solely reflects the intrinsic fidelity of the kernel approximation. Although any input $i(t)$ affects the system output $V(t)$ via the convolution in~\eqref{eq:conv_bessistor}
the kernel error \( \varepsilon(N) \) is defined without reference to \( i(t) \), making it a robust and input-agnostic metric for assessing the quality of the approximation.

Let us now consider the case $N=1$.
Then, the exact kernel \( z_B(t) \) is computed using~\eqref{eq:ZB-timedomain}, while the truncated kernel with $N = 1$ becomes
\begin{equation}
    z_B^{(1)}(t) = \frac{R_\infty}{\tau} \left[ 2(\nu+1) + 4(\nu+1)(\nu+2) e^{-j_{\nu+2,1} \, t/\tau} \right],
\end{equation}
where \( j_{\nu+2,1} \) is the first positive zero of the Bessel function \( J_{\nu+2}(x) \). The hybrid approximation adds the geometric tail \( G^{(1)}_\nu(t) \) defined in~\eqref{eq:GnuN}, leading to
\begin{equation}
z_B^{(1+\text{tail})}(t) = z_B^{(1)}(t) + \frac{R_\infty}{\tau} G^{(1)}_\nu(t)\,.    
\end{equation}

\section{Definition and properties of the Bessel element admittance}
\label{sec:admittance}

We are now going to consider the admittance related to the Bessel functions-based element that we are considering, namely
\begin{equation} \label{eq:Y_Bdef}
    \wt{Y}_B(s) = \frac{1}{\wt{Z}_B(s)} = 
    \frac{1}{R_\infty} \frac{I_{\nu+2}(\sqrt{s\tau})}{I_{\nu}(\sqrt{s\tau})} \,,\qquad \nu > -1 \, .
\end{equation}
This transition is not merely algebraic, but it offers a complementary perspective that is particularly advantageous in the analysis of parallel circuit configurations and in modeling systems where current is the primary observable.
For completeness, let analyze the asymptotic behavior of Bessel admittance, namely
\begin{gather}
    \wt{Y}_B (s) \overset{{s\to 0}}{\sim} \frac{\Gamma(\nu+1)}{\Gamma(\nu+3)} \frac{s\tau}{4\,R_\infty} \,, \label{eq:asymzeroY} \\
    \wt{Y}_B (s) \overset{{s\to \infty}}{\sim} \frac{1}{R_\infty} \,,\label{eq:asyminftyY}
\end{gather}
as intuitively deducible from the asymptotics of the impedance~\eqref{eq:asymzeroZ} and~\eqref{eq:asyminftyZ}.

\subsection{Mittag-Leffler Expansion of the Admittance}

As widely commented before, modified Bessel functions $I_\nu(z)$ are entire in $z \in \mathbb{C}$. Therefore, $I_\nu(\sqrt{s\tau})$ is analytic in $s$ except along a branch cut of the square root, and its set of zeros generates the simple poles of $\wt{Y}(s)$. The quotient $\wt{Y}(s)$ is thus a meromorphic function with simple poles at the points $ \displaystyle s_k = {z_k^2}/{\tau} $ where $z_k$ satisfies $I_\nu(z_k) = 0$.
Since the zeros of $I_\nu(z)$ lie on the imaginary axis, the real poles of $\wt{Y}(s)$ are located at:
\begin{equation} \label{eq:Yzeroesres}
    s_k = -\frac{j_{\nu,k}^2}{\tau} \quad \Rightarrow \quad \lambda_k = \frac{j_{\nu,k}^2}{\tau}
\end{equation}
where $j_{\nu,k}$ is the $k$-th positive zero of the classical Bessel function $J_\nu$.

So, let us employ Mittag-Leffler theorem~\cite{gorenflo2020mittag}, namely since $\wt{Y}(s)$ in~\eqref{eq:Y_Bdef} is meromorphic with simple poles $\{-\lambda_k\}$ and decays sufficiently fast at infinity, it can be expanded in a Mittag-Leffler series
\begin{equation}
\wt{Y}_B(s) = C+\sum_{k=1}^{\infty} \frac{a_k}{s + \lambda_k}\,,    
\end{equation}
leading to the result
\begin{equation}\label{eq:Yseries}
    \wt{Y}_B(s) = \frac{1}{R_\infty} - \frac{1}{R_\infty} \sum_{k=1}^\infty \frac{2 j_{\nu,k}}{\tau s + {j_{\nu,k}^2}} \frac{J_{\nu + 2} (j_{\nu,k})}{J'_{\nu} (j_{\nu,k})} \,, \qquad \nu > -1 \,,
\end{equation}
in terms of Bessel functions of first kind $J_\alpha (z)$ and their first derivative.
\begin{proof}
To prove~\eqref{eq:Yseries}, we first evaluate the constant $C$, that corresponds to the value of the function for $s\to \infty$. 
Subsequently, we compute the residue at the pole $s_k = -\lambda_k $ as in~\eqref{eq:Yzeroesres} given by
\begin{equation}
    a_k = \lim_{s \to -\lambda_k} (s + \lambda_k) \wt{Y}(s) \,.
\end{equation}    
and, for the computation, we the change the variable $z_k = \jmath j_{\nu,k}$, so that
\begin{equation}
a_k = \lim_{z \to z_k} \left( \frac{z^2 - z_k^2}{\tau} \right) \, \frac{1}{R_\infty} \, \frac{I_{\nu+2}(z)}{I_\nu(z)} \,.   
\end{equation}
Using L’H\^opital’s rule, we notice that this latter limit can be written as
\begin{equation}
a_k = \frac{1}{R_\infty} \, \frac{2 z_k}{\tau} \,\frac{I_{\nu+2}(z_k)}{I_\nu'(z_k)}    \,,
\end{equation}
which substituted in the initial expression for the admittance in the point $z_k = \jmath j_{\nu,k}$ turns out to be
\begin{equation}
    \wt{Y}_B(s) = \frac{1}{R_\infty}+\frac{1}{R_\infty} \sum_{k=1}^\infty \frac{2 \jmath j_{\nu,k}}{\tau} \, \frac{I_{\nu+2}(\jmath j_{\nu,k})}{I_\nu'(\jmath j_{\nu,k})} \, \frac{1}{s + \frac{j_{\nu,k}^2}{\tau}} \,,
\end{equation}
or, slighly simplified, 
\begin{equation} 
    \wt{Y}_B(s) = \frac{1}{R_\infty} + \frac{1}{R_\infty}\sum_{k=1}^\infty \frac{2 \jmath j_{\nu,k}}{\tau s + {j_{\nu,k}^2}} \, \frac{I_{\nu+2}(\jmath j_{\nu,k})}{I_\nu'(\jmath j_{\nu,k})} \,  \,.
\end{equation}
In conclusion, we then consider the following relations between Bessel functions and modified Bessel functions and their derivatives
\begin{equation}
    I_\nu (\jmath z) = \jmath^{-\nu} J_\nu(z)
 \,, \qquad     I^{(n)}_\nu (\jmath z) = \jmath^{n-\nu} J^{(n)}_\nu(z) \,,
 \end{equation}
leading to~\eqref{eq:Yseries}.
\end{proof}

The series in~\eqref{eq:Yseries} is a Stieltjes-type expansion and it is a valid and meaningful expansion in simple fractions. It is particularly interesting due to its physical interpretability and being an infinite series  that can be truncated to a finite number of terms.
In fact, it is particularly convenient in contexts as systems theory or viscoelastic modeling to refer to~\eqref{eq:Yseries} as Prony series and to approximate such functions using a sum of finite rational terms
\begin{equation}
    \wt{Y}_B(s) = \frac{1}{R_\infty} - \frac{1}{R_\infty} \sum_{k=1}^N \frac{2 j_{\nu,k}}{\tau s + {j_{\nu,k}^2}} \frac{J_{\nu + 2} (j_{\nu,k})}{J'_{\nu} (j_{\nu,k})} \,, \qquad \nu > -1 \,.
\end{equation}

\section{Numerical plots and applications}
\label{sec:plots}

\subsection{Characteristics of Bessel impedance}

Let us now analyze the impedance of such element, based on its qualitative behavior, by computing the magnitude and the phase of the complex impedance~\cite{Hayt2018engineering}.
These are usually obtained by setting first $s=\jmath \omega$, where $\omega$ is the pulsation in radians per second $(rad/\mathrm{s})$. For our purposes, it is convenient to express instead quantities in terms of the frequency $f$ in Hertz $(Hz)$, by means of $\displaystyle \omega = 2\pi f$.
The magnitude $\vert \wt{Z}_B (f) \vert$ is then obtained by the computation of the absolute value of the complex impedance, then expressed in decibels $(dB)$, i.~e.
\begin{equation}\label{eq:genmagZB}
\big\vert \wt{Z}_B(f) \big\vert = 20 \log_{10}\left( \sqrt{\Re\{\wt{Z}_B(2\pi\jmath f)\}^2 + \Im\{\wt{Z}_B(2\pi\jmath f)\}^2 }  \right)
\end{equation}
while the phase of such impedance, represented as $\angle \wt{Z}_B$, is instead defined according to 
\begin{equation}\label{eq:genphaseZB}
    \angle \wt{Z}_{B}(f) = \arctan\left(\frac{\Im\{\wt{Z}_{B}(2\pi\jmath f)\}}{\Re\{\wt{Z}_{B}(2\pi\jmath f)\}}\right) \,.
\end{equation}
From the technical point of view, it is also possible to write an explicit expression taking advantage of the techniques employed in the evaluation of the attenuation factor for Bessel models~\cite[Sec.~3]{colombaro2023Qbessel}, formalizing the $I_\alpha(z)$ in terms of Kelvin functions~\cite{Abramowitz1965handbook} $\ber{\alpha}{z}$ and $\bei{\alpha}{z}$ defined respectively as real and imaginary parts of the Bessel function of the first kind $J_{\alpha}\left(x e^{i\frac{3}{4}\pi}\right)$, namely
\begin{equation}
\label{eq:ber}
\displaystyle
\ber{\alpha}{z} := \left(\frac{z}{2}\right)^\alpha  \sum_{k=0}^\infty
\frac{\cos\left[\left( \frac{3\alpha}{4}+\frac{k}{2}\right)\pi\right]}{k!\Gamma(k+\alpha+1)} \left(\frac{z^2}{4}\right)^k \, ,
\end{equation}
\begin{equation}
\label{eq:bei}
\displaystyle
\bei{\alpha}{z} := \left(\frac{z}{2}\right)^\alpha  \sum_{k=0}^\infty
\frac{\sin\left[\left( \frac{3\alpha}{4}+\frac{k}{2}\right)\pi\right]}{k!\Gamma(k+\alpha+1)} \left(\frac{z^2}{4}\right)^k \, .
\end{equation}
As a result, we find
\begin{align}
    \Re\{\wt{Z}_B (2\pi\jmath f)\} &= 
    \frac{2}{\pi f}
        \frac{\bei{\nu+2}{\sqrt{2\pi f}}\bei{\nu}{\sqrt{2\pi f}}+\ber{\nu+2}{\sqrt{2\pi f}}\ber{\nu}{\sqrt{2\pi f}}}{\bei{\nu+2}{\sqrt{2\pi f}}^2+\ber{\nu+2}{\sqrt{2\pi f}}^2}   
    \,,\\
    \Im\{\wt{Z}_B (2\pi\jmath f)\} &= 
    \frac{2}{\pi f}
\frac{\bei{\nu+2}{\sqrt{2\pi f}}\ber{\nu}{\sqrt{2\pi f}}-\ber{\nu+2}{\sqrt{2\pi f}}\bei{\nu}{\sqrt{2\pi f}}}{\bei{\nu+2}{\sqrt{2\pi f}}^2+\ber{\nu+2}{\sqrt{2\pi f}}^2}  
    \,,
\end{align}
so that~\eqref{eq:genmagZB} and~\eqref{eq:genphaseZB} admits explicit expressions, in particular the phase is simplified to
\begin{equation}
   \angle \wt{Z}_{B}(f) = 
\arctan
        \left(
\frac{\bei{\nu+2}{\sqrt{2\pi f}}\ber{\nu}{\sqrt{2\pi f}}-\bei{\nu}{\sqrt{2\pi f}}\ber{\nu+2}{\sqrt{2\pi f}}}{\bei{\nu}{\sqrt{2\pi f}}\bei{\nu+2}{\sqrt{2\pi f}}+\ber{\nu}{\sqrt{2\pi f}}\ber{\nu+2}{\sqrt{2\pi f}}}
        \right) \,.
\end{equation}

Thus, the qualitative behavior for the magnitude and the phase as functions of the frequency is analyzed in~\figurename~\ref{fig:ZB}, fixing two among the three parameters $\displaystyle\nu$, $\displaystyle\tau$ and $\displaystyle R_\infty$ and for different values of the third one.
\begin{figure}[ht]
    \centering    \includegraphics[width=\linewidth]{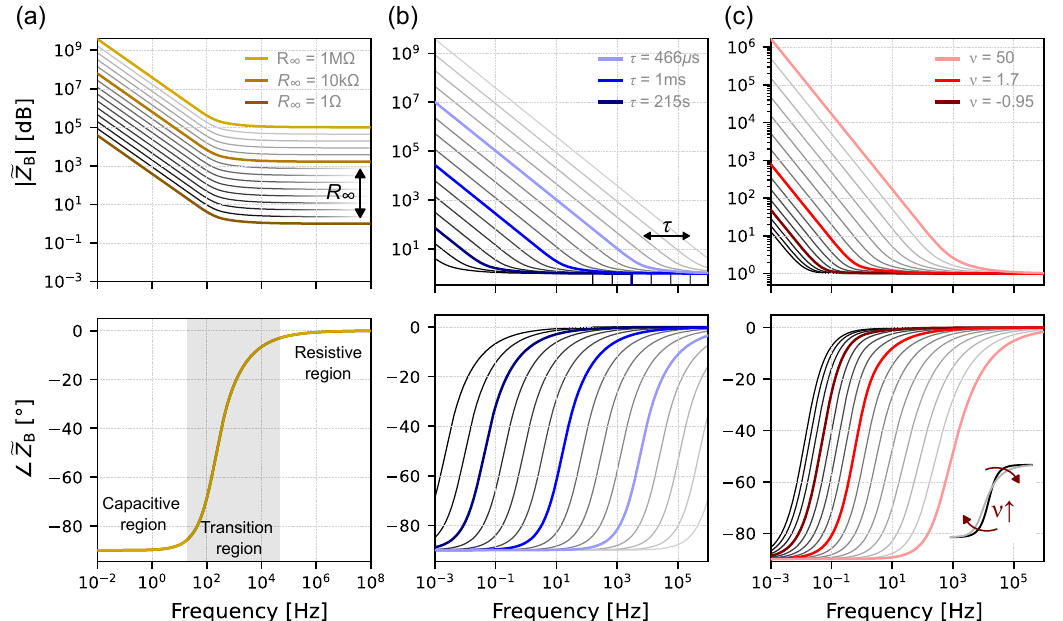}
\caption{Characteristics of the impedance $\wt{Z}_B(f)$, on the top the magnitudes $| \widetilde{Z}_B|$ expressed in $dB$ and on the bottom the phases $\angle \wt{Z}_B$ in grades. 
\textbf{(a)} Fixed $\nu = 1$ and $\tau = 10\, \mathrm{ms}$, varying $R_\infty$.
\textbf{(b)} Fixed $\nu = 1$ and $R_\infty = 1\,\Omega $, varying $\tau$. 
\textbf{(c)}
For varying $\nu$, including the limit value $\nu=-1$, fixed $R_\infty = 1\, \Omega$ and $\tau = 1 \, \mathrm{s}$.}
    \label{fig:ZB}
\end{figure}
Interestingly, we can also appreciate how $\nu$ is related to to the memory.
It is noteworthy that the impedance exhibits capacitive behavior at low frequencies, as indicated by the high magnitude. At intermediate frequencies, the magnitude decreases, signaling a transition toward resistive behavior, while at high frequencies, it tends to stabilize.
An interesting consequence of this behavior might be done by considering the comparison with an RC circuit, recovered for the limit value of $\nu=-1$, as might be appreciate in~\figurename~\ref{fig:aligned_nu}. Tuning the parameter $\nu$  induces noticeable changes in both the magnitude and phase profiles. Notably, the characterization of the Bessel element involves contributions from both resistive and capacitive effects.
\begin{figure}[h!]
    \centering    \includegraphics[width=\linewidth]{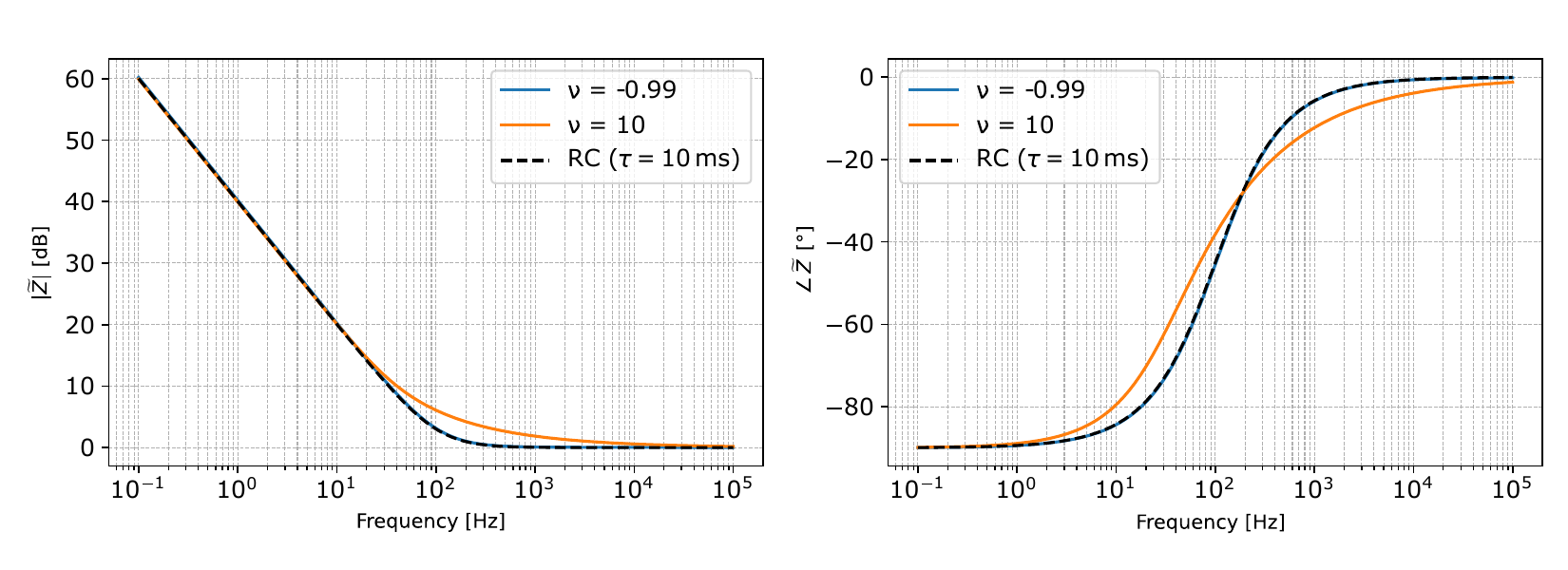}
\caption{Comparison in frequency response between a passive RC circuit ($\tau = \text{10\,ms}$) and $\widetilde{Z}_B(f)$ with $R_{\infty} = 1\,\Omega$, for $\nu = -0.99$ ($\tau = 2\pi \cdot 10\,\mu\text{s}$) and $\nu = 10$ ($\tau = 0.85\,\text{s}$).}
    \label{fig:aligned_nu}
\end{figure}

\subsection{Circuits configuration}

The transition from capacitive to resistive behavior reflects the dielectric relaxation processes active at different spatial and temporal scales, and constitutes a fundamental feature of several physical phenomena.
As discussed in the previous section, the Bessel element impedance reproduces the characteristic relaxation response of dispersive systems, exhibiting a smooth transition from a capacitive regime at low frequency to a resistive plateau at high frequency.

The frequency-dependent transition from capacitive behavior at low frequencies to resistive behavior at high frequencies can be effectively captured by a minimal electrical model consisting of a Bessel element connected in parallel with a resistor $R_0$.
The resistor in parallel accounts for the DC conductivity of the system, ensuring that the model captures both the dynamic dielectric behavior and the baseline conductive path. This simple configuration provides a physically interpretable and analytically tractable foundation for modeling tissue impedance over a wide frequency range. Mathematically, the total impedance of the parallel configuration is simply given by
\begin{equation}\label{eq:ZRparB}
    \wt{Z}_{R_0||B}(s) = \frac{R_0 \wt{Z}_B(s)}{R_0 + \wt{Z}_B(s)} \,.
\end{equation}
Concerning configuration in~\eqref{eq:ZRparB}, we find another noteworthy similarities with classical models used to describe dielectric relaxation. First, employing~\eqref{eq:asymzeroZ} and~\eqref{eq:asyminftyZ}, let us notice that the asymptotic limits at low and high frequencies are respectively
\begin{gather}
    \wt{Z}_{R_0||B}(s) \overset{{s\to 0}}{\sim} R_0   \,, \\ 
    \wt{Z}_{R_0||B}(s) \overset{{s\to \infty}}{\sim} \frac{R_0 R_\infty}{R_0 + R_\infty} \,. \label{eq:ZRB-sinfty}
\end{gather}
So, at low frequencies there is match of the low cutoff frequencies resistor $R_0 $ in the parallel circuit and the behavior of the classical models of relaxation, namely Cole-Cole, Davidson-Cole and Havriliak–Negami~\cite{garrappa2016modelsofdielectric}.
At high frequencies, the value of the impedance is stabilizing at a certain value, that is not necessarily vanishing, differently from classical dielectric models.

To capture the full spectral complexity of some physical models, as might be biological tissues, a single relaxation element is often insufficient \cite{gabriel1996iii}. A more comprehensive representation can be achieved by extending the minimal model to include multiple Bessel elements in parallel, each accounting for a distinct relaxation mechanism active in a specific frequency range, all connected in parallel with the DC resistor $R_0$, thereby forming a modular and physically interpretable model capable of reproducing the layered dielectric response of biological media. Each impedance $\wt{Z}_{B,i}$ introduces an independent set of parameters \((R_{\infty,i}, \tau_i, \nu_i)\), enabling precise control over the position, shape, and strength of each dispersion region. as a consequence, the total impedance of the system ${\wt{Z}_{\text{tot}}}$ is then given by
\begin{equation}\label{eq:ZRparBN}
    \frac{1}{\wt{Z}_{\text{tot}}(s)} = \frac{1}{R_0} + \sum_{i=1}^{N} \frac{1}{\wt{Z}_{B,i}(s)} \,.
\end{equation}

\begin{figure}[h]
    \centering    \includegraphics[width=\linewidth]{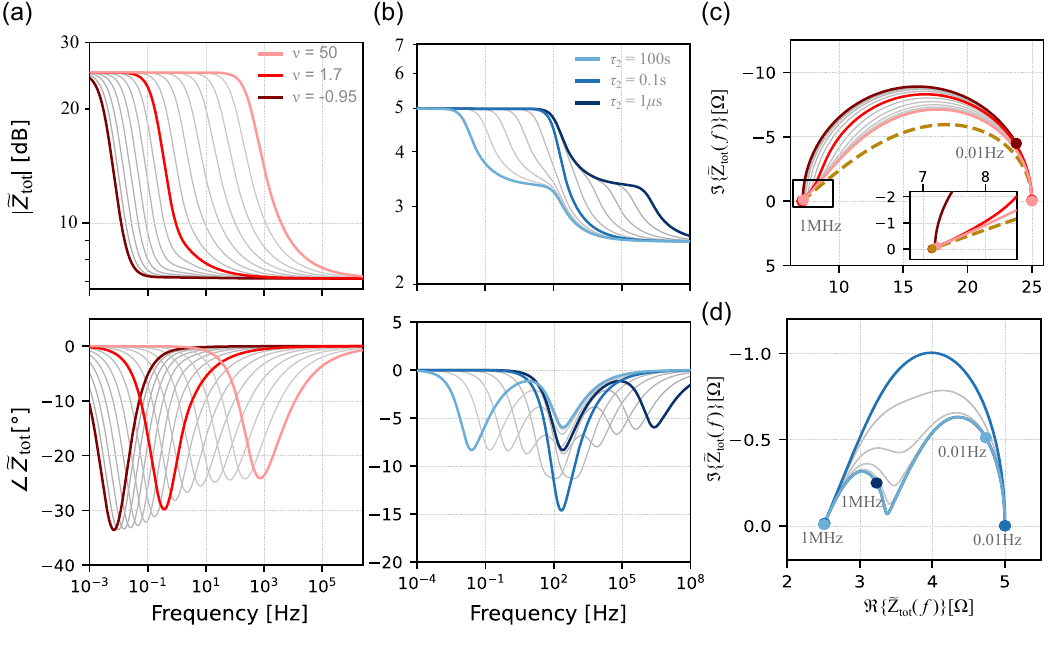}
    \caption{
        Frequency response of $\widetilde{Z}_{\mathrm{tot}}$ in~\eqref{eq:ZRparBN}.
        \textbf{(a)} Magnitude and phase of $\widetilde{Z}_{\text{tot}}$ for $N=1$, with $R_0 = 25\, \Omega$, $R_\infty = 10\,\Omega$ and $\tau = 1 \, \mathrm{s}$.
        \textbf{(b)} Magnitude and phase of $\widetilde{Z}_{\text{tot}}$ for $N=2$, for $R_0= 5\,\Omega$, $R_{\infty,1}=R_{\infty,2}=10\,\Omega$, $\nu_1 = \nu_2 =1$, $\tau_1 = 10 \, \mathrm{ms}$.
        \textbf{(c)} Nyquist plot for $N=1$ with values in \textbf{(a)}.
        \textbf{(d)} Nyquist plot for $N=2$ with values in \textbf{(b)}.
    }
    \label{fig:results_RBB}
\end{figure}
In Figure~\ref{fig:results_RBB},
the frequency response of the configuration $\widetilde{Z}_{\text{tot}}$ defined in~\eqref{eq:ZRparBN} is shown for the cases $N = 1$ and $N=2$.
For $N=1$, the results exhibit the characteristic relaxation or dispersion behavior of the system for fixed \( R_0 \), \( R_{\infty} \), \( \tau\) and varying \( \nu\). At low frequencies, the impedance is dominated by the resistor \( R_0 \), since the Bessel element behaves as a high-impedance capacitive element in this regime, as noticed before. As the frequency increases, a smooth transition occurs in which the relaxation becomes more evident and the impedance decreases accordingly. At high frequencies, the response converges to the equivalent impedance of \( R_0 \parallel R_{\infty} \), in agreement with the asymptotic behavior in~\eqref{eq:ZRB-sinfty}.

The parameter \( \nu \) plays a central role in controlling the sharpness of this capacitive-to-resistive transition. As we may appreciate in the Nyquist plot in~\figurename~\ref{fig:results_RBB}\textbf{(c)}, when \( \nu \) is close to \(-1\), the response closely resembles that of a classical RC circuit, with an almost ideal semicircular arc. As \( \nu \) increases, the transition becomes more gradual, and the Nyquist trajectory flattens and becomes increasingly deformed. Notably, for large values of $\nu$ (e.g., $\nu = 50$), the high-frequency segment of the Nyquist plot becomes approximately linear. This indicates that the impedance exhibits a quasi-constant ratio between its imaginary and real components over a broad frequency range, corresponding to a nearly constant phase angle. Such behavior is characteristic of systems with structured, multiscale dispersion \cite{Mainardi2010book, hilfer2000}.
It is also important to note that, as previously observed in the standalone Bessel element case, there exists a strong coupling between the parameters \( \nu \) and \( \tau \). Although \( \tau \) is fixed in the case $N=1$, varying \( \nu \) shifts the frequency at which the transition from capacitive to resistive behavior occurs.
This effect arises because \( \nu \) modulates the spectral density of the relaxation modes associated with the Bessel element and, interestingly, higher values of \( \nu \) tend to concentrate the response within a narrower frequency band, while lower values spread the transition across a broader range.
Moreover, this connection between the parameters \( \nu \) and \( \tau \) provides an additional degree of freedom for shaping the spectral profile, allowing fine-tuned modeling of dispersive systems without altering the characteristic time constant.

Let us consider instead~\eqref{eq:ZRparBN} for $N = 2$. With respect to the previous case, we now vary the time parameter \( \tau_2 \), while keeping fixed \( \tau_1 = 10\,\mathrm{ms} \) and setting the remaining parameters \( R_0 = 5\,\Omega \), \( R_{\infty,1} = R_{\infty,2} = 10\,\Omega \), and \( \nu_1 = \nu_2 = 1 \).
The results clearly demonstrate that when the time constants \( \tau_1 \) and \( \tau_2 \) are well separated, the impedance response exhibits two distinct relaxation mechanisms. This is evident in the Bode plots of both magnitude and phase in~\figurename~\ref{fig:results_RBB}\textbf{(b)}, where two separate transitions can be observed, while from the Nyquist diagram in~\figurename~\ref{fig:results_RBB}\textbf{(d)}, the presence of two semicircular arcs further confirms the occurrence of two independent dispersions.
As the values of \( \tau_1 \) and \( \tau_2 \) become closer, the relaxation processes begin to merge, resulting in a broader and more pronounced single transition. The Nyquist plot illustrates how the two arcs gradually collapse into a single, deformed semicircle. This behavior can be interpreted as an additive superposition of the individual Bessel impedance responses. When the time parameters are sufficiently close, the overall response resembles that of a single relaxation mechanism with increased intensity and modified curvature, in this case shaped by the shared value \( \nu = 1 \).

\subsection{Validation with biological systems}
In order to validate the theoretical model, we now analyze and compare its characteristics with respect to some existing models of biological systems.
Specifically, in~\figurename~\ref{fig:fittingdata}, we compare the impedance magnitude and phase with the data coming from dry skin and muscle tissue, fitting Bessel functions-based model frequency impedance for a 1-cm square tissue cube \cite{andreuccetti2012}.
\begin{figure}[ht]
    \centering
\includegraphics[width=\linewidth]{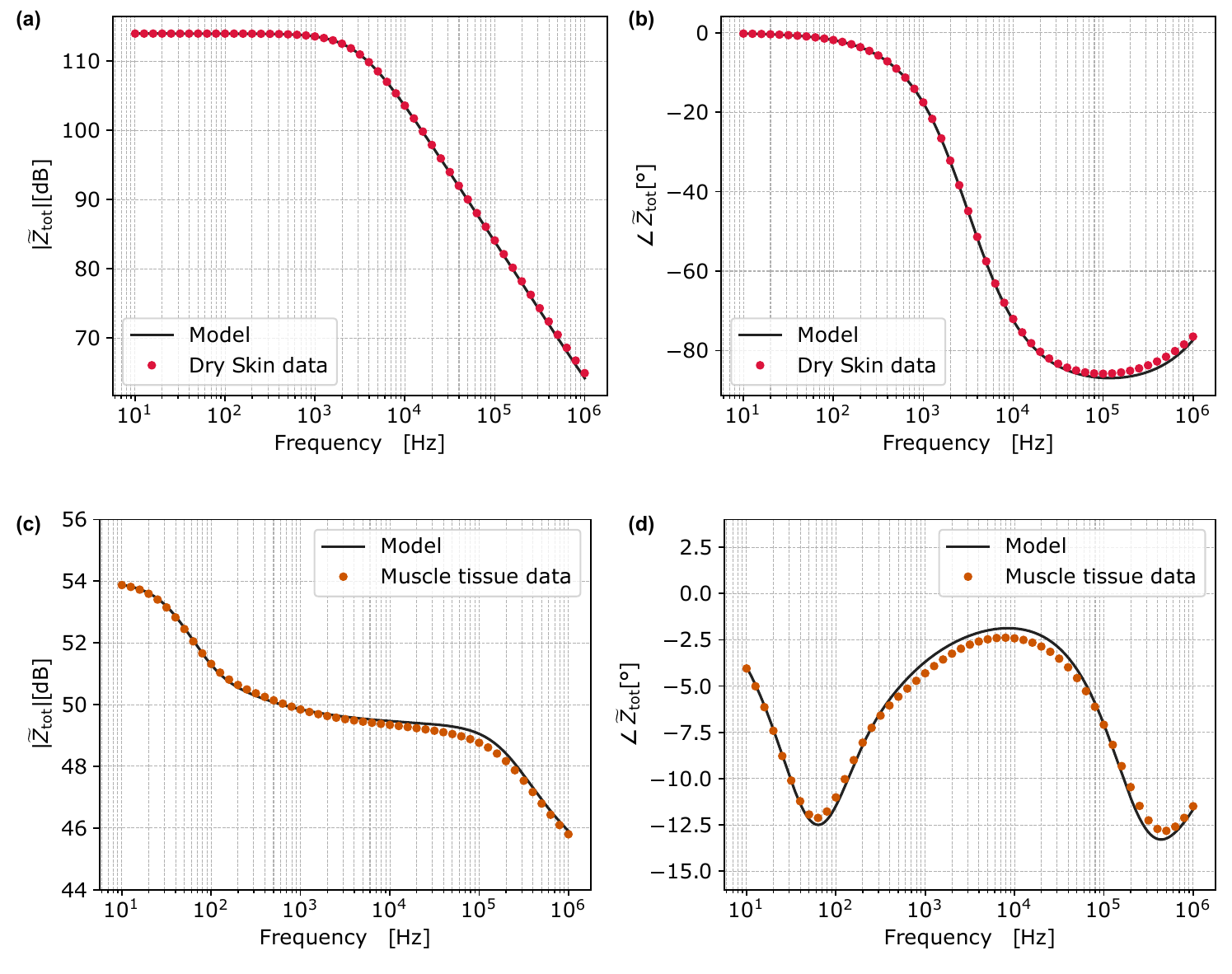}
    \caption{In \textbf{(a)} and \textbf{(b)}, we have the comparison of the proposed Bessel impedance magnitude $| \widetilde{Z}_{\text{tot}}|$ and phase $\angle \widetilde{Z}_{\text{tot}}$ and the data about dry skin.
    In \textbf{(c)} and \textbf{(d)}, we have the evaluation of the Bessel element behavior with respect to the data of the muscle tissue.}
    \label{fig:fittingdata}
\end{figure}
To fit the experimental data of dry skin, we consider the general impedance having $R_0=500 \, k\Omega$ in parallel with a Bessel element, namely $N=1$, with the parameters $R_{\infty}=190\,\Omega$, $\tau=12\,\mu\mathrm{s}$, $\nu =0.5 $.
Concerning the muscle tissue, we have instead two Bessel elements in parallel with
$R_0= 50\, \Omega$, while the other parameters are respectively $R_{\infty,1}=700\,\Omega$, $\tau_1=10\,\mathrm{ms}$, $\nu_1 = -0.15 $ and $R_{\infty,2}=300\,\Omega$, $\tau_2=35\,\mu\mathrm{s}$, $\nu_2 =5 $.
Looking at~\figurename~\ref{fig:fittingdata}, we notice that the model exhibits a remarkable qualitative consistency with the data, highlighting the robustness and well-posedness of the model, concerning tissues modeling.

\section{Discussion}
\label{sec:discussion}

Classical models used to describe spectral relaxation, such as Cole--Cole, Havriliak--Negami, and fractional-order elements, are widely adopted due to their flexibility and compact functional forms. However, these models rely on empirical or nonlocal formulations that introduce several limitations. Fractional derivatives, in particular, lead to infinite-memory operators, complicating time-domain analysis, simulation, and physical interpretation. Their parameters, such as the dispersion exponent~$\alpha$, are often introduced as fitting constants without a clear connection to underlying physical mechanisms. Moreover, these models typically lack closed-form time-domain solutions and are not straightforward to implement in circuit-level simulations or hardware systems~\cite{sabatier2007}.

Electro-mechanical analogy with viscoelastic Bessel models addresses these challenges by introducing a structurally grounded alternative based on modified Bessel functions. Bessel impedance arises naturally from the solution of an underlying differential relation, with the parameter $\nu$ embedded directly as the order of the Bessel function. This formulation ensures internal consistency and enables the model to satisfy key physical and mathematical properties such as analyticity, passivity, BIBO stability, and monotonicity.

From a computational perspective, this novel Bessel element enables direct implementation in both frequency and time domains using closed-form expressions. Unlike fractional-order models, it does not require auxiliary states, rational approximations, or kernel fitting, when expressed in the frequency domain.
Its impulse response is given by an infinite series of exponentially decaying terms with analytically defined time constants, derived from the zeros of the Bessel function. This modal expansion converges rapidly and can be truncated adaptively without loss of stability or causality. As a result, it can be integrated into numerical solvers and circuit-level simulators as a passive and memoryless block, eliminating the need for infinite-memory convolution kernels or nonlocal operators. 

Crucially, the proposed Bessel passive element reproduces essential features of fractional-order systems such as power-law decay and broadband dispersion, without relying on fractional calculus. This provides a consistent and physically interpretable mechanism for modeling distributed relaxation phenomena using realizable components. The model effectively bridges the gap between phenomenological flexibility and structural interpretability. Furthermore, future investigations might explore extensions to more complex dielectric systems, which have been shown to follow generalized fractional dynamics~\cite{giusti2018prabhakar,giusti2020practicalguide,giusti2020generalPrabhakar, colombaro2024electromech-prabhakar}.
Although this Bessel-based element retains a clear structural and physical interpretation, we note that multiple parameter combinations can yield comparable fits to experimental data. This is a known characteristic of inverse problems in broadband dispersive systems, including classical models such as Cole--Cole. A systematic exploration of parameter sensitivity and identifiability, potentially through hierarchical or constrained fitting strategies, lies beyond the scope of the present study but represents a relevant direction for future work.

\section{Conclusions}
\label{sec:end}

This work introduces a novel passive circuit element whose impedance is defined by a closed-form expression involving modified Bessel functions of the first kind. The model provides a physically interpretable and mathematically rigorous alternative to classical relaxation models, avoiding the drawbacks commonly associated with fractional or empirical formulations.

The proposed element satisfies essential properties for realistic system modeling, including analyticity, passivity, BIBO stability, and monotonicity. Its formulation captures both capacitive and resistive behavior, with smooth spectral transitions governed by a small number of interpretable parameters. In contrast to fractional-order elements, the Bessel element reproduces power-law behavior and broadband dispersion without invoking nonlocal operators or infinite-memory kernels.
Its structure and behavior are applicable to any context involving multiscale relaxation dynamics, including dielectric materials, viscoelastic media, biological systems, acoustically lossy systems, and complex soft-matter environments. In these domains, it can serve as a drop-in replacement for fractional-order models that lack physical realizability.

Overall, the proposed passive element depending on Bessel functions establishes a structured and physically grounded framework for modeling distributed relaxation. It bridges the gap between analytical tractability and practical applicability in complex systems and it constitutes a foundational contribution toward a new generation of modeling tools. Its compatibility with both simulation and physical implementation makes it suitable for widespread use in fields where broadband relaxation is a defining feature.

\section*{Acknowledgments}

The authors would like to thank the anonymous reviewer for their valuable comments and suggestions, which helped improve the clarity and quality of this manuscript.
The authors are grateful to Dr.~Aida Villa\'ecija for valuable comments and support. 
The work of I.~C. has been carried out in the framework of the activities of the Italian National Group of Mathematical Physics (GNFM), INdAM. 


\end{document}